\documentclass[conference]{IEEEtran}
\usepackage{amsmath,amssymb}
\usepackage[dvips]{graphicx}
\usepackage{amsfonts}
\usepackage[mathscr]{eucal}
\usepackage{latexsym}
\usepackage{amsthm}
\usepackage{exscale}
\usepackage[mathscr]{eucal}
\usepackage{bm}
\usepackage[dvipsnames]{color}
\usepackage{cases}
\usepackage{epsfig}
\usepackage[center,small]{caption}
\usepackage{algorithm}
\usepackage{algorithmic}
\usepackage[verbose,nospace,sort]{cite}
\usepackage{tabularx}
\scrollmode

\IEEEoverridecommandlockouts

\newtheorem{proposition}{Proposition}
\newtheorem{definition}{Definition}
\newtheorem{theorem}{Theorem}

\newcommand{\xlog}{\mathrm{log}}

\newcommand{\bC}{\mathbf{C}}
\newcommand{\xC}{\widetilde{C}}
\newcommand{\xbC}{\widetilde{\mathbf{C}}}
\newcommand{\lbC}{\underline{\mathbf{C}}}
\newcommand{\xE}{\mathbb{E}}

\begin{document}
\title{Improving Achievable Rate for the Two-User SISO Interference Channel with Improper Gaussian Signaling}
\author{\authorblockN{Yong~Zeng$^{\dag}$, Cenk M. Yetis$^{\dag}$, Erry Gunawan$^{\dag}$, Yong~Liang~Guan$^{\dag}$, and Rui~Zhang$^{\ddag}$\\}
\authorblockA{${\dag}$School of Electrical and Electronic Engineering, Nanyang Technological University, Singapore\\
${\ddag}$ECE Department, National University of Singapore\\
ze0003ng@e.ntu.edu.sg, \{cenkmyetis, egunawan, eylgua\}@ntu.edu.sg, elezhang@nus.edu.sg}
}
\IEEEspecialpapernotice{(Invited Paper)}
\maketitle

\begin{abstract}
This paper studies the achievable rate region of the two-user single-input-single-output (SISO) Gaussian interference channel, when the \emph{improper} Gaussian signaling is applied. Under the assumption that the interference is treated as additive Gaussian noise, we show that the user's achievable rate can be expressed as a summation of the rate achievable  by the conventional \emph{proper} Gaussian signaling, which depends on the users' input \emph{covariances} only, and an additional term, which is a function of both the users' covariances and \emph{pseudo-covariances}. The additional degree of freedom given by the pseudo-covariance, which is conventionally set to be zero for the case of proper Gaussian signaling, provides an opportunity to improve the achievable rate by employing the improper Gaussian signaling. Since finding the optimal solution for the joint covariance and pseudo-covariance optimization is difficult, we propose a sub-optimal but efficient algorithm by separately optimizing these two sets of parameters. Numerical results show that the proposed algorithm provides a close-to-optimal performance as compared to the exhaustive search method, and significantly outperforms the optimal proper Gaussian signaling and other existing improper Gaussian signaling schemes.
\end{abstract}


\section{Introduction}
The capacity of the two-user Gaussian interference channel (IC) has been an open problem for a long time \cite{164}. Recently, a significant progress has been made in \cite{163}, where it is proved that a particular Han-Kobayashi type scheme can achieve within one bit to the information-theoretical capacity. A key technique in the Han-Kobayashi scheme is to split each user's transmit signal into a common message, which is decodable at both receivers, and a private message, which is decodable at the intended receiver only. Since this  capacity-approaching technique requires signal-level encoding/decoding cooperations among the users, which are difficult to implement, a more pragmatic approach is to implement single-user detection at the receivers by treating the interference as noise \cite{233,240,243}.

There has been a great deal of research on characterizing the Pareto boundary of the achievable rate region for the Gaussian IC with the interference treated as noise \cite{240,250,243,249}. A Pareto boundary consists of all the rate-tuples at each of which it is impossible to improve one  particular user's rate, without simultaneously decreasing the rate of at least one of the other users. A common approach for such characterizations is via solving a sequence of weighted sum-rate maximization (WSRMax) problems \cite{243}. However, as pointed out in \cite{240}, the WSRMax approach  cannot guarantee the finding of all Pareto-boundary points due to the non-convexity of the achievable rate set. An alternative method based on the concept of \emph{rate profile} was proposed  in \cite{240}, which is able to characterize the complete Pareto boundary for the multiple-input-single-output IC (MISO-IC) . 
 Furthermore, the rate-profile approach generally results in optimization problems that  are easier to handle than the conventional WSRMax problems \cite{240}. 

However, all the aforementioned works are restricted to \emph{proper} Gaussian input signals, for which the  second-order statistic is completely specified by the covariance matrix (under the zero-mean assumption). On the other hand, for the more general \emph{improper} Gaussian signaling, an extra parameter called \emph{pseudo-covariance} is required for the complete second-order characterization of the complex-valued input signals \cite{245,246,248}. This extra parameter provides a new  opportunity to improve the achieve rates of the Gaussian ICs. For instance, it was shown in \cite{189} that improper signaling is beneficial in improving the degrees-of-freedoms (DoF) performance for the three-user single-input-single-output IC (SISO-IC)  with time-invariant channel coefficients. Since the DoF metric is  meaningful only at the asymptotically high signal-to-noise ratio (SNR), in this paper, we are interested in characterizing the achievable rate region with improper Gaussian signaling at any finite SNR value. For the purpose of exposition, we consider the simple two-user SISO-IC in this paper.

The achievable rate region for the two-user SISO-IC with improper Gaussian signaling has been studied in \cite{255,256}, based on the equivalent $2\times 2$ real-valued multiple-input-multiple-output (MIMO) channel obtained by separating the real and imaginary parts of the channel coefficients. In \cite{255}, a rank-1 signaling scheme was proposed, i.e., the transmit covariance matrices in the equivalent MIMO channel are restricted to be rank one. This is equivalent to transmit purely real or purely imaginary signals over the original complex-valued SISO channel. In \cite{256}, the rate region was obtained based on an exhaustive search over the input covariance matrices. It was shown that by exploiting  the symmetry property of the covariance matrices, the search space can be reduced to the $4$-dimension; however, the resulted search is still of high complexity.  In contrast to the above prior works, in this paper, we adopt  the complex-valued SISO channel model  to gain some new insights. 
 Based on existing results on improper complex random vectors (RVs),  we show  that the achievable rate with improper Gaussian signaling for the two-user SISO-IC can be expressed as a summation of the rate achievable by the conventional proper Gaussian signaling, which depends on the users' input covariances only, and an additional term, which is a function of both the users' covariances and pseudo-covariances. By applying the rate-profile technique, we then propose an efficient algorithm to optimize the covariances and pseudo-covariances to enlarge the  achievable rate region. 

The rest of this paper is organized as follows.
Section~\ref{S:systemModel} introduces the system model and some preliminaries.
Section~\ref{S:main} presents the problem formulation and the proposed algorithm. Numerical results are provided in Section~\ref{S:simulation}. Finally, we conclude the paper in Section~\ref{S:conclusions}.

\section{System Model}\label{S:systemModel}
Consider a two-user SISO-IC, where each transmitter is intended to send one independent message to the corresponding receiver. The input-output relationship is given by
\begin{align}
y_1&=h_{11}x_1+h_{12}x_2+n_1 , \\
y_2&=h_{21}x_1+h_{22}x_2+n_2 ,
\end{align}
where $y_1$ and $y_2$ are the received signals at receiver $1$ and $2$, respectively; $h_{rt}=|h_{rt}|e^{j\phi_{rt}}, r,t=1,2$ denotes the complex channel coefficient from transmitter $t$ to receiver $r$ and $\phi_{rt}$ represents its phase; $n_1$ and $n_2$ are the zero mean circularly symmetric complex Gaussian (CSCG) noises with identical variance $\sigma^2$, denoted by $n_1, n_2 \sim \mathcal{CN}(0, \sigma^2)$; and $x_1$ and $x_2$ are the independent signals from transmitter $1$ and $2$, respectively. Different from the conventional setup where proper Gaussian signaling is assumed, in this paper, $x_1$ and $x_2$ are zero-mean complex Gaussian random variables which can be \emph{improper}.

\subsection{Preliminary: Improper Random Vectors}
For a zero-mean RV $\mathbf{z}\in \mathbb{C}^{n}$, the covariance matrix $\bC_{\mathbf{z}}$ and the pseudo-covariance matrix $\xbC_z$ are defined as \cite{245}
\begin{align}\label{E:CDef}
\bC_{\mathbf{z}}&\triangleq \xE(\mathbf{z}\mathbf{z}^H),\  \xbC_{\mathbf{z}}\triangleq \xE(\mathbf{z}\mathbf{z}^T),
\end{align}
where $(\cdot)^T$  and  $(\cdot)^H$ represent the transpose and complex-conjugate transpose, respectively.

\begin{definition} \cite{245}:
 A complex RV $\mathbf{z}$ is called proper if its pseudo-covariance matrix $\xbC_{\mathbf{z}}$ vanishes to a zero matrix; otherwise it is called improper.
\end{definition}
Define $\lbC_{\mathbf{z}}$ as the covariance matrix of the augmented vector $[\begin{matrix} \mathbf{z}^T & \mathbf{z}^{*T} \end{matrix}]^T$, where $(\cdot)^*$ represents the complex conjugate operation, i.e.,
\begin{align}
\lbC_{\mathbf{z}}\triangleq \xE \bigg[\begin{matrix} \mathbf{z} \\ \mathbf{z}^* \end{matrix}\bigg] \bigg[\begin{matrix} \mathbf{z} \\ \mathbf{z}^* \end{matrix}\bigg]^H=\bigg[\begin{matrix}\bC_{\mathbf{z}} & \xbC_{\mathbf{z}} \\ \xbC_{\mathbf{z}}^* & \bC_{\mathbf{z}}^* \end{matrix} \bigg].
\end{align}
\begin{theorem} \cite{246}: \label{T:validPair}
$\bC_{\mathbf{z}}$ and $\xbC_{\mathbf{z}}$ are a valid pair of covariance and pseudo-covariance matrices, i.e., there exists a RV $\mathbf{z}$ with covariance and pseudo-covariance matrices given by $\bC_{\mathbf{z}}$ and $\xbC_{\mathbf{z}}$, respectively,  if and only if the augmented covariance matrix $\lbC_{\mathbf{z}}$ is positive semidefinite.
\end{theorem}
\begin{theorem} \label{T:entropyGaussian} \cite{246}:
 The differential entropy of a complex Gaussian RV $\mathbf{z}$ with augmented covariance matrix $\lbC_{\mathbf{z}}$ is given by
\begin{align}
h(\mathbf{z})=\frac{1}{2}\xlog[(\pi e)^{2n}\det\lbC_{\mathbf{z}}].
\end{align}
\end{theorem}

\subsection{Achievable Rate with Improper Gaussian Signaling}
For a scalar complex random variable $z$, we use $C_z$ and $\xC_z$ to denote  the covariance and pseudo-covariance, respectively. 
Then for the zero-mean input Gaussian signals $x_1$ and $x_2$, we have
\begin{align}
C_{x_1}&=\xE(x_1x_1^*), \ \xC_{x_1}=\xE(x_1x_1), \\
C_{x_2}&=\xE(x_2x_2^*), \ \xC_{x_2}=\xE(x_2x_2).
\end{align}
Note that $C_{x_1}$ and $C_{x_2}$ are nonnegative real numbers equal to the power values  of the transmitted signals, while $\xC_{x_1}$ and $\xC_{x_2}$ are complex numbers in general. From Theorem~\ref{T:validPair}, it is easy to verify that the following conditions are both necessary and sufficient for $C_{x_1}$ and $\xC_{x_1}$ (or $C_{x_2}$ and $\xC_{x_2}$) to be a valid pair of covariance and pseudo-covariance for a random variable $x_1$ (or $x_2$)
\begin{align}\label{E:validPair}
|\xC_{x_1}|\leq C_{x_1}, \  |\xC_{x_2}|\leq C_{x_2}.
\end{align}
 Next, we derive the rate expression in terms of the input covariances and pseudo-covariances. With Gaussian input $x_1$ and $x_2$, it is well known that the received signals $y_r, r=1, 2,$ are also Gaussian. The
covariance and pseudo-covariance of $y_r$ can be obtained as
\begin{align}
C_{y_r}& =\xE(y_r y_r^*)=|h_{r1}|^2C_{x_1}+|h_{r2}|^2C_{x_2}+\sigma^2, \notag \\
\xC_{y_r}&=\xE(y_r y_r)=h_{r1}^2\xC_{x_1}+h_{r2}^2\xC_{x_2}. \label{E:y1PC}
\end{align}
With the conditions given in \eqref{E:validPair}, it can be verified that $|\xC_{y_r}|^2<C_{y_r}^2$ is satisfied.
Then from Theorem~\ref{T:entropyGaussian}, the differential entropy of $y_r$ is given by
\begin{align}
h(y_r)&=\frac{1}{2}\xlog (\pi e)^{2}\det \bigg[\begin{matrix} C_{y_r} & \xC_{y_r} \\ \xC_{y_r}^* & C_{y_r} \end{matrix} \bigg] \notag \\
&=\xlog(\pi eC_{y_r})+\frac{1}{2}\xlog(1-C_{y_r}^{-2}|\xC_{y_r}|^2). \notag
\end{align}
Define $s_r= h_{r \overline{r}}x_{\overline{r}}+n_r$, which is the interference plus noise term at receiver $r$, where $\overline{r}=\mod(r+1, 2)$. Then
\begin{align}
C_{s_r}& =|h_{r \overline{r}}|^2C_{x_{\overline{r}}}+\sigma^2,\
\xC_{s_r}=h_{r \overline{r}}^2\xC_{x_{\overline{r}}}. \label{E:s1PC}
\end{align}
Similarly, the differential entropy of $s_r$ can be obtained as
\begin{align}
h(s_r)=\xlog(\pi eC_{s_r})+\frac{1}{2}\xlog(1-C_{s_r}^{-2}|\xC_{s_r}|^2).
\end{align}
Under the assumption that interference is treated as additive Gaussian noise and perfect channel knowledge is known  at all terminals, the achievable rate at receiver $r$ with improper Gaussian signaling can be obtained as
\begin{align}
R_r&=I(x_r;y_r)
=h(y_r)-h(s_r)\notag \\
&=\underbrace{\xlog \big(1+\frac{|h_{rr}|^2C_{x_r}}{\sigma^2+|h_{r \overline{r}}|^2C_{x_{\overline{r}}}}\big)}_{R_r^{\text{proper}}(C_{x_1}, C_{x_2})}+\frac{1}{2}\xlog\frac{1-C_{y_r}^{-2}|\xC_{y_r}|^2}{1-C_{s_r}^{-2}|\xC_{s_r}|^2}.
\label{E:R1}
\end{align}
Equation \eqref{E:R1} clearly shows that, compared to the conventional proper Gaussian signaling, the achievable rate with improper Gaussian signaling has an additional term, which is a function of both the covariances and pseudo-covariances of the input signals.  By setting $\xC_{x_1}$ and $\xC_{x_2}$ both to be $0$, \eqref{E:R1} reduces to the well-known rate expression for the proper Gaussian signaling. 
\section{Achievable Rate Region with Improper Gaussian Signaling}\label{S:main}
The achievable rate region for the two-user IC is defined to be the set of rate-pairs for both users that can be simultaneously achieved under a given set of transmit power constraints for each transmitter, denoted by $P_1, P_2$, i.e.,:
\begin{align}
\mathcal{R}\triangleq \bigcup_{\begin{subarray}{l} C_{x_1}\leq P_1,|\xC_{x_1}|\leq C_{x_1}\\C_{x_2}\leq P_2, |\xC_{x_2}|\leq C_{x_2}\end{subarray}} \big\{(r_1, r_2): 0\leq r_1 \leq R_1, 0\leq r_2 \leq R_2 \big\},\notag
\end{align}
where $R_1$ and $R_2$ are given by \eqref{E:R1}. 
\begin{definition}  \cite{249}:
A rate-pair $(r_1,r_2)$ is Pareto optimal if there is no other rate-pair $(r_1',r_2')$ with $(r_1',r_2')\geq (r_1,r_2)$ and $(r_1',r_2')\neq (r_1,r_2)$, where the inequality is component-wise.
\end{definition}
For the achievable rate region $\mathcal{R}$ with improper Gaussian signaling, we adopt the rate-profile technique in \cite{240} to characterize the Pareto optimal rate-pairs. Specifically, any Pareto optimal rate-pair of $\mathcal{R}$ can be obtained by solving the following optimization problem with a particular rate profile denoted by  $(\alpha, 1-\alpha)$:
\begin{align}
\text{(P1):} & \underset{C_{x_1},C_{x_2},\xC_{x_1},\xC_{x_2}, R}{\text{max.}}   R \notag \\
\text{s.t.} \quad & R_1^{\text{proper}}(C_{x_1},C_{x_2})+\frac{1}{2}\xlog\frac{1-C_{y_1}^{-2}|\xC_{y_1}|^2}{1-C_{s_1}^{-2}|\xC_{s_1}|^2}\geq \alpha R , \notag \\
& R_2^{\text{proper}}(C_{x_1},C_{x_2})+\frac{1}{2}\xlog\frac{1-C_{y_2}^{-2}|\xC_{y_2}|^2}{1-C_{s_2}^{-2}|\xC_{s_2}|^2}\geq (1-\alpha)R , \notag \\
&0\leq C_{x_1} \leq P_1, \ 0\leq C_{x_2}\leq P_2, \notag \\
& |\xC_{x_1}|^2\leq C_{x_1}^2, \  |\xC_{x_2}|^2\leq C_{x_2}^2, \notag 
\end{align}
where $\alpha \in [0,1]$ denotes the target ratio between user 1's achievable rate and the users' sum-rate, $R$. Denote the optimal solution to (P1) as $R^{\star}$, then the rate-pair $( \alpha R^{\star}, (1-\alpha)R^{\star})$ must be on the Pareto boundary corresponding to the rate profile given by $\alpha$. Thereby, by solving (P1) with different $\alpha$ values between $0$ and $1$, the complete Pareto boundary for the achievable rate region $\mathcal{R}$ can be found. However, solving (P1) by jointly optimizing the covariances and pseudo-covariances is quite involved. We therefore propose a suboptimal solution in this paper with separate optimization of the covariances and pseudo-covariances. Specifically, with pseudo-covariances set to zeros, (P1) is shown to reduce to a convex  optimization problem, which can be efficiently solved with linear programming (LP). On the other hand, for fixed covariances, (P1) is shown to be equivalent to solving a finite  number of second-order cone programming (SOCP) problems, from which the pseudo-covariances can be optimally solved efficiently. 

\subsection{Covariance Optimization}\label{S:covarianceOpt}
When restricted to proper Gaussian signaling with $\xC_{x_1}=0$ and $\xC_{x_2}=0$, (P1) reduces to
\begin{align}
\text{(P1-a):}\quad &\underset{r,C_{x_1},C_{x_2}}{\text{max.}} \quad  r \notag \\
\text{s.t.}\quad & \xlog (1+\frac{|h_{11}|^2C_{x_1}}{\sigma^2+|h_{12}|^2C_{x_2}})\geq  \alpha r, \label{C:r1P1} \\
& \xlog (1+\frac{|h_{22}|^2C_{x_2}}{\sigma^2+|h_{21}|^2C_{x_1}})\geq (1-\alpha)r, \label{C:r2P1}  \\
&0\leq C_{x_1} \leq P_1, \ 0\leq C_{x_2}\leq P_2. \notag
\end{align}
(P1-a) is non-convex and hence cannot be solved directly. However, for any fixed value $r$, (P1-a) can be transformed to the following LP feasibility problem:
\begin{align}
\text{(P1-a'):}& \quad \text{Find} \quad  C_{x_1},C_{x_2} \notag \\
\text{s.t.} \quad & |h_{11}|^2C_{x_1}\geq (\sigma^2+|h_{12}|^2C_{x_2})(e^{\alpha r}-1),  \notag \\
& |h_{22}|^2C_{x_2}\geq (\sigma^2+|h_{21}|^2C_{x_1})(e^{(1-\alpha)r}-1), \notag \\
&0\leq C_{x_1} \leq P_1, \ 0\leq C_{x_2}\leq P_2. \notag
\end{align}
(P1-a') can be efficiently solved with existing algorithms such as the simplex method \cite{260}. If $r$ is feasible to (P1-a'), then it follows that the optimal solution to (P1-a) satisfies $r^{\star}\geq r$; otherwise, $r^{\star} < r$. Thus, (P1-a) can be efficiently solved by solving (P1-a') with different values for $r$, together with the bisection method for updating $r$\cite{202}.
Denoting the optimal solution to (P1-a) as $\{r^{\star}, C_{x_1}^{\star}, C_{x_2}^{\star}\}$, then it can be verified that the constraints \eqref{C:r1P1} and \eqref{C:r2P1} will be both active, i.e.,
\begin{align}\label{E:equalr}
R_1^{\text{proper}}(C_{x_1}^{\star}, C_{x_2}^{\star})=\alpha r^{\star}, R_2^{\text{proper}}(C_{x_1}^{\star}, C_{x_2}^{\star})=(1-\alpha)r^{\star}.
\end{align}

\subsection{Pseudo-Covariance Optimization}
In this subsection, (P1) is optimized over the pseudo-covariances $\xC_{x_1}$ and $\xC_{x_2}$, by fixing the covariances  to $C_{x_1}^{\star}$ and $C_{x_2}^{\star}$ obtained by solving (P1-a). The resulted problem is formulated as
\begin{align}
\text{(P1-b):} &\quad \underset{\xC_{x_1},\xC_{x_2}, R}{\text{max.}} \quad  R \notag \\
\text{s.t.} \quad & R_1^{\text{proper}}(C_{x_1}^{\star},C_{x_2}^{\star})+\frac{1}{2}\xlog\frac{1-C_{y_1}^{-2}|\xC_{y_1}|^2}{1-C_{s_1}^{-2}|\xC_{s_1}|^2}\geq \alpha R , \notag \\
& R_2^{\text{proper}}(C_{x_1}^{\star},C_{x_2}^{\star})+\frac{1}{2}\xlog\frac{1-C_{y_2}^{-2}|\xC_{y_2}|^2}{1-C_{s_2}^{-2}|\xC_{s_2}|^2}\geq (1-\alpha)R,  \notag \\
& |\xC_{x_1}|^2\leq C_{x_1}^{\star 2}, \  |\xC_{x_2}|^2\leq C_{x_2}^{\star 2} \notag,
\end{align}
where $C_{y_1}, C_{s_1}, C_{y_2}$ and $C_{s_2}$ are the corresponding covariance terms with  input covariances $C_{x_1}^{\star}$ and $C_{x_2}^{\star}$.
 Again, if a given $R$ is achievable for certain $\xC_{x_1}, \xC_{x_2}$, then the optimal solution to (P1-b) satisfies $R^{\star}\geq R$; otherwise, $R^{\star} < R$. This enables solving (P1-b) via solving a set of feasibility problems, each for a fixed value $R$. Moreover, $R$ can be updated  with a simple bisection search \cite{202}. Substituting \eqref{E:equalr} into (P1-b), it can be easily obtained that $\{\xC_{x_1}=0, \xC_{x_2}=0, R=r^{\star}\}$ is feasible to (P1-b). Therefore,  $R^{\star}\geq r^{\star}$ is satisfied, i.e., with our proposed separate covariance and pseudo-covariance optimization, the sum-rate corresponding to the rate profile given by $\alpha$ with improper Gaussian signaling is no smaller than that obtained with the optimal proper Gaussian signaling. 

Next, we present the algorithm for the feasibility problem resulting from (P1-b) for a given target $R$. Substituting \eqref{E:y1PC}, \eqref{E:s1PC} and \eqref{E:equalr} into (P1-b) and after some simple manipulations, the feasibility problem for a given $R$ can be formulated as
\begin{align}
\text{(P1-b')}\ & \underset{\xC_{x_1},\xC_{x_2}}{\text{min.}} \  0 \notag \\
\text{s.t.} \quad & a_1|h_{11}^2\xC_{x_1}+h_{12}^2\xC_{x_2}|^2+b_1\leq |\xC_{x_2}|^2,\label{C:ineq1} \\
& a_2|h_{21}^2\xC_{x_1}+h_{22}^2\xC_{x_2}|^2+b_2\leq |\xC_{x_1}|^2, \label{C:ineq2} \\
& |\xC_{x_1}|^2\leq C_{x_1}^{\star 2}, \label{C:ineq3} \\
 & |\xC_{x_2}|^2\leq C_{x_2}^{\star 2},\label{C:ineq4}
\end{align}
\begin{align}
\text{where } a_1&=\frac{C_{s_1}^2}{\beta_1 C_{y_1}^2 |h_{12}|^4}, \ &&b_1=\frac{(1-1/\beta_1)C_{s_1}^2}{|h_{12}|^4}, \notag \\
a_2&=\frac{C_{s_2}^2}{\beta_2 C_{y_2}^2 |h_{21}|^4}, \ &&b_2=\frac{(1-1/\beta_2)C_{s_2}^2}{|h_{21}|^4}, \notag \\
\beta_1&=e^{2\alpha(R-r^{\star})}, &&\beta_2=e^{2(1-\alpha)(R- r^{\star})}.\notag
\end{align}
Since $R^{\star}\geq r^{\star}$, we can assume that $R\geq r^{\star}$ without loss of optimality. Then it follows that $\beta_1\geq 1, \beta_2\geq 1, b_1\geq 0$ and $b_2\geq 0$.

\begin{figure}
\centering
\includegraphics[width=3.3in, height=2.5in]{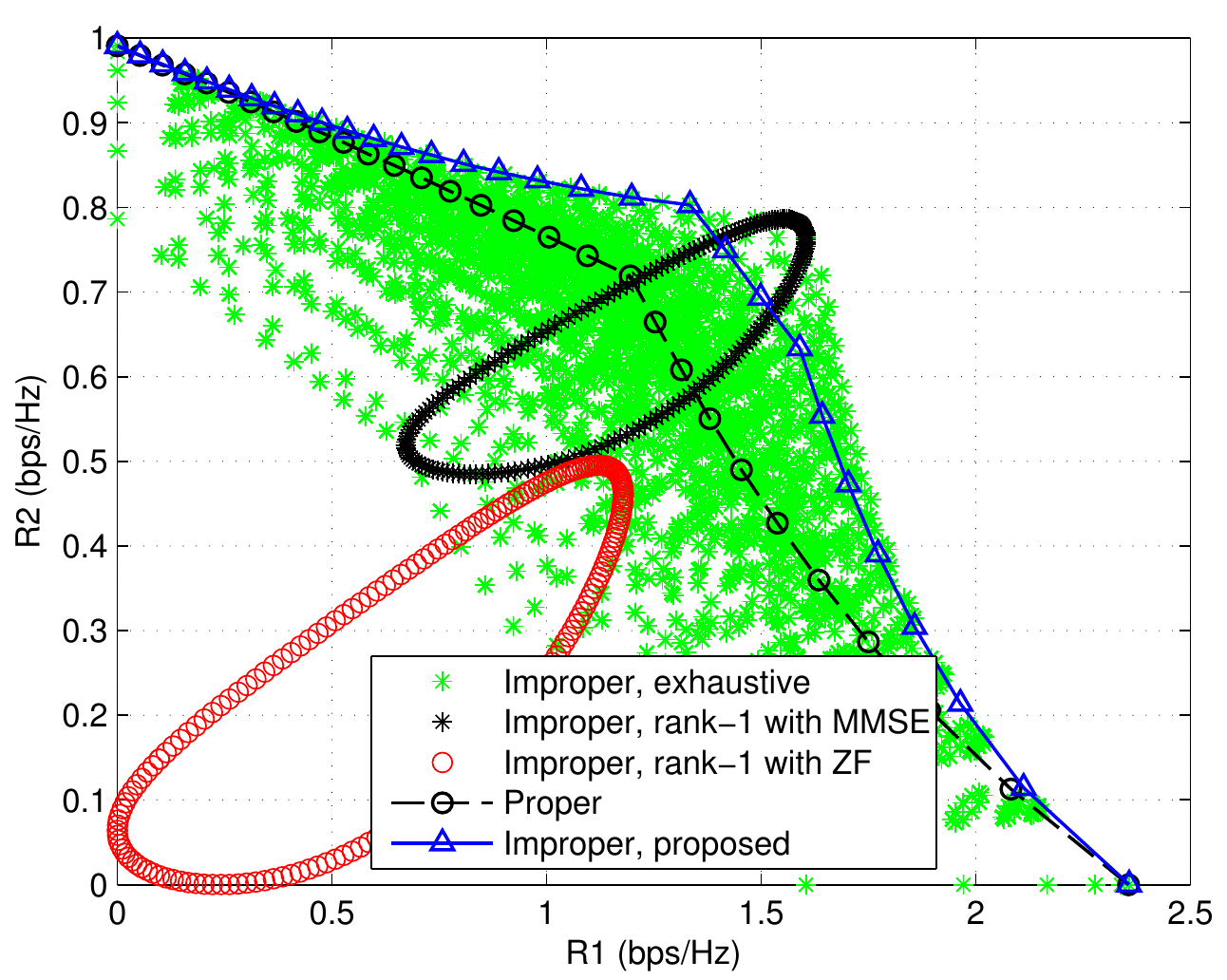}
\caption{Achievable rate region for SNR = 0 dB}\label{F:region0dB}
\end{figure}

(P1-b') is non-convex and hence cannot be solved directly. Next, we show that it can be efficiently solved via solving a finite number of SOCP problems. First, observe that an arbitrary common phase rotation can be added to both pseudo-covariances without affecting the feasibility of (P1-b'). That is, if $\{\xC_{x_1},\xC_{x_2}\}$ is feasible for (P1-b'), then so is  $\{\xC_{x_1}e^{j\omega},\xC_{x_2}e^{j\omega}\}$. Therefore, without loss of generality, we may choose $\omega$ so that $\xC_{x_1}$ is real and nonnegative.
 Denote the magnitude and phase of $\xC_{x_2}$ by $t$ and $\theta$, respectively, i.e., $\xC_{x_2}=te^{j\theta}$. Then for any fixed value of $\theta$, (P1-b') can be transformed into a SOCP feasibility problem given by
\begin{align}
\text{(P1-b''):} \quad \underset{\xC_{x_1},t}{\text{min.}} \quad & 0 \notag \\
\text{s.t.} \quad & \bigg\|\begin{matrix} \sqrt{a_1}(h_{11}^2\Re\{\xC_{x_1}\}+h_{12}^2te^{j\theta}) \\ \sqrt{b_1} \end{matrix} \bigg\| \leq t, \notag \\
& \bigg\|\begin{matrix} \sqrt{a_2}(h_{21}^2\Re\{\xC_{x_1}\}+h_{22}^2te^{j\theta}) \\ \sqrt{b_2} \end{matrix} \bigg\| \leq \Re\{\xC_{x_1}\}, \notag \\
& \Im\{\xC_{x_1}\}=0,\  \Re\{\xC_{x_1}\}\leq C_{x_1}^{\star}, \ t\leq C_{x_2}^{\star}. \notag
\end{align}
(P1-b'') is  convex and hence can be efficiently solved with the standard interior point algorithm \cite{202}, or by existing software tools such as \texttt{CVX} \cite{227}.
\begin{theorem}\label{T:thetaRegion}
The feasibility problem (P1-b') can be optimally solved by solving a finite number of  SOCP problems (P1-b''), each for a fixed value $\theta$, where $\theta$ can be restricted to the following discrete set:
\begin{align}
\theta \in \{\pi+2(\phi_{11}-\phi_{12}), \pi+2(\phi_{21}-\phi_{22})\}\cup \Theta_{\mathcal{A}}\cup \Theta_{\mathcal{B}}, \notag
\end{align}
where $\Theta_{\mathcal{A}}$ and $\Theta_{\mathcal{B}}$ are the solution sets for $\theta$ to the following equations:
 \begin{align}
 \Theta_{\mathcal{A}}:
\begin{cases}\label{E:thetaA}
a_1|h_{11}^2C_{x_1}^{\star}+h_{12}^2te^{j\theta}|^2+b_1=t^2 \\
a_2|h_{21}^2C_{x_1}^{\star}+h_{22}^2te^{j\theta}|^2+b_2=C_{x_1}^{\star 2}\\
\end{cases}
\end{align}
\begin{align}
\Theta_{\mathcal{B}}:
\begin{cases}\label{E:thetaB}
a_1|h_{11}^2X_1+h_{12}^2C_{x_2}^{\star}e^{j\theta}|^2+b_1=C_{x_2}^{\star 2} \\
a_2|h_{21}^2X_1+h_{22}^2C_{x_2}^{\star}e^{j\theta}|^2+b_2=X_1^{ 2}\\
\end{cases}
\end{align}
\end{theorem}
\begin{proof}
Please refer to Appendix~\ref{A:thetaRegion}
\end{proof}

Theorem~\ref{T:thetaRegion} can be intuitively interpreted as follows. For the feasibility problem (P1-b'), if the constraint \eqref{C:ineq1} is more ``restrictive'' than \eqref{C:ineq2}, then $\theta$ should have a value such that the left hand side of \eqref{C:ineq1} is minimized. This corresponds to $\theta=\pi+2(\phi_{11}-\phi_{12})$ so that $h_{11}^2\xC_{x_1}$ and $h_{12}^2\xC_{x_2}$ are antiphase. Similar interpretation for $\theta=\pi+2(\phi_{21}-\phi_{22})$ can be made. If both \eqref{C:ineq1} and \eqref{C:ineq2} are equally ``restrictive'', a feasible solution tends to make both constraints satisfied with equality, as given by \eqref{E:thetaA} and \eqref{E:thetaB}. 
The elements in $\Theta_{\mathcal{A}}$ and $\Theta_{\mathcal{B}}$ can be easily obtained as shown in Appendix~\ref{A:thetas}. 

\section{Numerical Results}\label{S:simulation}
In Fig.~\ref{F:region0dB} and Fig.~\ref{F:region10dB}, the achievable rate regions (prior to any time-sharing of achievable rate-pairs) of the proposed improper Gaussian signaling scheme at SNR=$0$ versus  $10$ dB are compared with other schemes, including the optimal proper Gaussian signaling scheme presented in Section~\ref{S:covarianceOpt}, the  optimal improper Gaussian signaling obtained with the exhaustive search method \cite{256}, and the  rank-1 scheme with both zero-forcing (ZF) and minimum-mean-square error (MMSE) beamforming \cite{255}. 
The channel matrix for both plots is given by $\mathbf{H}=\Big[\begin{matrix} 1.5718 - 1.2863i & -1.2984 + 0.7032i\\
  -0.2847 + 0.6700i  & 0.7802 - 0.6151i \end{matrix} \Big]$. The $(r,t)$th element of $\mathbf{H}$ is $h_{rt}$, which is the channel coefficient from transmitter $t$ to receiver $r$. 
  Both figures reveal that the achievable rate regions have been significantly enlarged with improper Gaussian signaling. It is also observed that our proposed algorithm, albeit being sub-optimal due to the separate optimization of covariances and pseudo-covariances,  performs quite close to the optimal improper Gaussian signaling by the exhaustive search. Another interesting observation is that for this particular channel realization, the Pareto boundary points of the achievable rate region with time-sharing (by taking the convex-hull operation over all the achievable rate-pairs without time-sharing) with improper Gaussian signaling can be obtained by the time-sharing between the two single-user maximum rate points, and the north-east rate corner point of the rank-1 scheme with MMSE beamforming.   
\begin{figure}
\centering
\includegraphics[width=3.3in, height=2.5in]{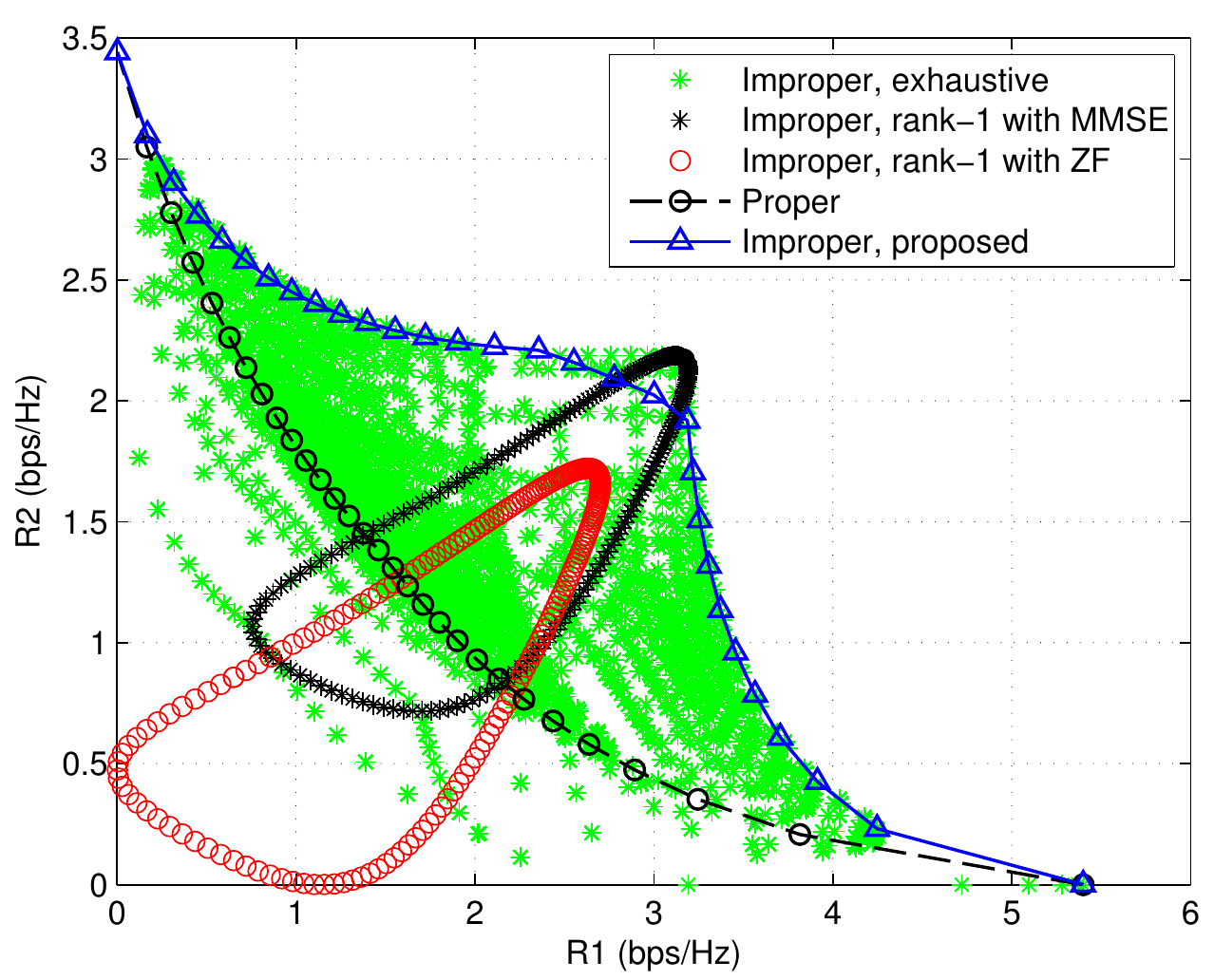}
\caption{Achievable rate region for SNR = 10 dB}\label{F:region10dB}
\end{figure}
\section{Conclusion}\label{S:conclusions}
This paper studies the achievable rate region of the two-user SISO Gaussian IC, when the improper Gaussian signaling is applied. 
 It is shown that the achievable rate can be expressed as a summation of the conventional rate expression by  proper Gaussian signaling, and an additional term that is a function of both the input covariances and pseudo-covariances.
 An efficient algorithm is proposed to obtain enlarged achievable rate regions by optimizing the covariances and pseudo-covariances separately with improper Gaussian signalling.
\appendices
\section{Proof of Theorem~\ref{T:thetaRegion}}\label{A:thetaRegion}
For notational convenience, in the sequel, we use $X_1$ and $X_2$ to denote $\xC_{x_1}$ and $\xC_{x_2}$, respectively. First, the following proposition shows that to solve (P1-b'), we may consider exterior solutions only, i.e., the solutions at which at least one of the inequality constraints is active.
\begin{proposition}\label{Pr:active}
If $\{X_1,X_2\}$ is feasible to (P1-b') with $|X_1|<C_{x_1}^{\star}$ and $|X_2|<C_{x_2}^{\star}$, then there exists another feasible solution $\{X_1', X_2'\}$ with $|X_1'|=C_{x_1}^{\star}$ or $|X_2'|=C_{x_2}^{\star}$.
\end{proposition}
\begin{proof}
Let $\tau\triangleq \min\big\{\frac{C_{x_1}^{\star}}{|X_1|}, \frac{C_{x_2}^{\star}}{|X_2|}\big\}$. Then $\tau >1$. Define $X_1'=\tau X_1, X_2' =\tau X_2$. Then it is easy to verify that the constraints in \eqref{C:ineq3} and \eqref{C:ineq4} are satisfied, i.e., $|X_1'|\leq C_{x_1}^{\star}$ and $|X_2'|\leq C_{x_2}^{\star}$. Furthermore, at least one of them are satisfied with equality. The constraint in \eqref{C:ineq1} is also satisfied since
\begin{align}
&a_1|h_{11}^2X_1'+h_{12}^2X_2'|^2+b_1
=\tau^2a_1|h_{11}^2X_1+h_{12}^2X_2|^2+b_1 \notag \\
&\overset{(a)}{\leq} \tau^2 (a_1|h_{11}^2X_1+h_{12}^2X_2|^2+b_1)\overset{(b)}{\leq}  \tau^2|X_2|^2=|X_2'|^2, \notag
\end{align}
where $(a)$ is satisfied since $\tau>1$ and $b_1\geq 0$, $(b)$ is true since $\{X_1, X_2\}$ is feasible to (P1-b').
Similarly, \eqref{C:ineq2} is also satisfied.  Therefore, $\{X_1', X_2'\}$ is a feasible solution to (P1-b') with at least one of the inequality constraints being active.
\end{proof}

  Next, we derive Theorem~\ref{T:thetaRegion} using the KKT conditions, which are necessary optimality conditions for the constrained optimization problem (P1-b') \cite{202}. For notational convenience, denote the inequality constraints \eqref{C:ineq1}--\eqref{C:ineq4} by $f_1\leq 0, f_2\leq 0, h_1\leq 0$ and $ h_2\leq 0$, respectively.  Denote $\lambda_1, \lambda_2, \mu_1, \mu_2$ as the corresponding dual variables, respectively. Then the Lagrangian function of (P1-b') is given by
  \begin{align}
L(&X_1,  X_2, \lambda_1, \lambda_2, \mu_1, \mu_2)=\lambda_1 \Big\{a_1 \big(|h_{11}|^4X_1^*X_1+|h_{12}|^4X_2^*X_2 \notag \\
 &+2\Re \{h_{11}^{2*}h_{12}^2X_1^*X_2\}\big) +b_1-X_2^*X_2\Big\}
+\lambda_2\Big\{a_2 \big(|h_{21}|^4X_1^*X_1 \notag \\
&+|h_{22}|^4X_2^*X_2
+2\Re \{h_{21}^{2*}h_{22}^2X_1^*X_2\}\big)
+b_2-X_1^*X_1\Big\} \notag \\
&+\mu_1(X_1^*X_1-C_{x_1}^{\star 2}) +\mu_2(X_2^*X_2-C_{x_2}^{\star 2}). \label{E:Lagrangian}
\end{align}

For $\{X_1, X_2\}$ to be a solution to (P1-b'), the following KKT conditions must be satisfied \cite{202}: 
\begin{enumerate}
\item{Dual feasibility: $\lambda_1\geq 0, \  \lambda_2\geq 0, \ \mu_1\geq 0, \  \mu_2\geq 0.$
}
\item{Zero derivative: The derivatives of the Lagrangian function \eqref{E:Lagrangian} with respect to the primal variables are zero:
    \begin{align}
    \frac{\partial L}{\partial X_2^*}=0 \Rightarrow & -X_2 \underbrace{\big [ (a_1|h_{12}|^4-1)\lambda_1+\lambda_2a_2|h_{22}|^4+\mu_2 \big]}_{c_2}\notag \\ &=X_1(\lambda_1\underbrace {a_1h_{12}^{2*}h_{11}^2}_{V_1}+\lambda_2\underbrace{a_2h_{22}^{2*}h_{21}^2}_{V_2})\notag
    \end{align}
    \begin{align}
\frac{\partial L}{\partial X_1^*}=0 \Rightarrow & X_1 \underbrace{\big [ (a_2|h_{21}|^4-1)\lambda_2+\lambda_1a_1|h_{11}|^4+\mu_1 \big]}_{c_1}\notag \\ &=-X_2(\lambda_1\underbrace {a_1h_{11}^{2*}h_{12}^2}_{V_1^*}+\lambda_2\underbrace{a_2h_{21}^{2*}h_{22}^2}_{V_2^*}),\notag \\
& \Downarrow \notag \\
 -X_2c_2&=X_1(\lambda_1V_1+\lambda_2V_2), \label{E:X2X1}\\
 X_1c_1&=-X_2(\lambda_1V_1^*+\lambda_2V_2^*). \label{E:X2X1Another}
    \end{align}
    }
    \item{Complementary slackness: $\lambda_1f_1= 0,  \lambda_2 f_2= 0, \mu_1 h_1= 0, \mu_2 h_2= 0.$ 
        }
\end{enumerate}
As discussed previously, without loss of generality, $X_1$ can be assumed to be a nonnegative real number. The following cases are then considered to derive the possible phases of $X_2$:
\begin{itemize}
\item{Case I: $f_1=0$ and $f_2\neq 0$. Then from the complementary slackness condition, $\lambda_1>0$ and $\lambda_2=0$. Substituting them into \eqref{E:X2X1Another}, we have
    \[ X_2=-\frac{X_1(\lambda_1a_1|h_{11}|^4+\mu_1)}{\lambda_1|V_1|^2}V_1.\]
     Since $\lambda_1\geq 0$, $\mu_1\geq 0$, $a_1\geq 0$ and $X_1\geq 0$, the phase $\theta$ of $X_2$ equals to that of $V_1$ rotated by $\pi$, which is $\pi+2(\phi_{11}-\phi_{12})$ since $V_1=a_1h_{12}^{2*}h_{11}^2$.}
\item{Case II: $f_1\neq 0$ and $f_2= 0$. Then $\lambda_1=0$ and $\lambda_2>0$. 
     Similarly, by using \eqref{E:X2X1}, 
     we have $\theta=\pi+2(\phi_{21}-\phi_{22})$.}
\item{Case III: $f_1\neq 0$ and $f_2 \neq 0$, then $\lambda_1=0$ and $\lambda_2=0$. By substituting them into \eqref{E:X2X1} and \eqref{E:X2X1Another}, we have $\mu_2X_2=0$ and $\mu_1X_1=0$. This means either $X_2=0$, $X_1=0$, or $\mu_1=0$, $\mu_2=0$. The former case is trivially corresponding to proper Gaussian signaling. Furthermore, Proposition~\ref{Pr:active} suggests that we may consider the exterior solutions only, i.e., either $h_1=0$ or $h_2=0$ is satisfied. Thus, $\mu_1>0$ or $\mu_2>0$ can be assumed. Therefore, case III can be ignored without loss of optimality.}
\item{Case IV: $f_1=0$ and $f_2=0$, then $\lambda_1>0$ and $\lambda_2>0$. In this case, $\theta$ belongs to the solution set for the equations given in \eqref{E:thetaA} and \eqref{E:thetaB}, which are obtained by satisfying the constraints in \eqref{C:ineq1} and \eqref{C:ineq2} with equality. \eqref{E:thetaA} and \eqref{E:thetaB} correspond to $|X_1|=C_{x_1}^{\star}$ and $|X_2|=C_{x_2}^{\star}$, respectively, which can be assumed without loss of generality due to Proposition~\ref{Pr:active}.}
\end{itemize}
This completes the proof of Theorem~\ref{T:thetaRegion}.
\section{Solving $\Theta_{\mathcal{A}}$ and $\Theta_{\mathcal{B}}$ in Theorem~\ref{T:thetaRegion}}\label{A:thetas}
In this appendix, we show the steps to solve  $\Theta_{\mathcal{A}}$. $\Theta_{\mathcal{B}}$ can be obtained similarly. The unknown variables in \eqref{E:thetaA} are $\theta$ and $t$. 
 After some manipulations, \eqref{E:thetaA} can be written as
\begin{align}
 t &\cos \eta + d_1t^2+d_2=0 \label{E:eq1}\\
 t &\cos(\eta+\omega)+d_3t^2+d_4=0 \label{E:eq2}
 \end{align}
 where
 \begin{align}
&\omega \triangleq 2(\phi_{22}+\phi_{11}-\phi_{12}-\phi_{21}), \ \eta \triangleq \theta+2(\phi_{12}-\phi_{11}) \label{E:thetaEta}\\
&d_1\triangleq \frac{a_1|h_{12}|^4-1}{2a_1|h_{11}|^2|h_{12}|^2C_{x_1}^{\star}}, \
d_2\triangleq \frac{a_1|h_{11}|^4C_{x_1}^{\star 2}+b_1}{2a_1|h_{11}|^2|h_{12}|^2C_{x_1}^{\star}}, \notag \\
&d_3\triangleq \frac{|h_{22}|^2}{2|h_{21}|^2C_{x_1}^{\star}},\qquad
d_4=\frac{(a_2|h_{21}|^4-1)C_{x_1}^{\star 2}+b_2}{2a_2|h_{21}|^2|h_{22}|^2C_{x_1}^{\star}}.\notag
\end{align}
From \eqref{E:eq2}, we have
\begin{align}
& t\sin\eta\sin\omega=t\cos\eta\cos\omega+d_3t^2+d_4 \Rightarrow \notag \\
&  t^2(1-\cos^2\eta)\sin^2\omega=(t\cos\eta\cos\omega+d_3t^2+d_4)^2 \label{E:eq3}
\end{align}
Solving $\cos \eta$ from \eqref{E:eq1}, we have
\begin{align}
\cos \eta=-(d_1t^2+d_2)/t. \label{E:eq4}
\end{align}
Substituting \eqref{E:eq4} into \eqref{E:eq3} gives the following fourth order polynomial equation with respect to $t$:
\begin{align}
 [t^2 -(d_1t^2+d_2)^2]\sin^2\omega=[(d_3-d_1\cos\omega)t^2+d_4-d_2\cos\omega]^2\notag
\end{align}
Since the above equation only has $t^2$ terms, it can be transformed to the following quadratic equation by setting $z=t^2$,
\begin{align}
&e_1z^2+e_2z+e_3=0,\\
\text{where }
e_1=&d_3^2+d_1^2-2d_1d_3\cos\omega, \notag \\
e_2=&2(d_1d_2+d_3d_4)-2(d_1d_4+d_2d_3)\cos\omega-\sin^2\omega,\notag \\
e_3=&d_2^2+d_4^2-2d_2d_4\cos\omega.\notag
\end{align}
Then $z$ can be easily solved. Since $z=t^2$ and $t$ is the magnitude of $\xC_{x_2}$, only the solutions of $z$ that are real and satisfy $0 \leq z \leq C_{x_2}^{\star 2}$ need to be kept, whereby the values for $t$ are obtained. For those values of $t$ satisfying $|(d_1t^2+d_2)/t|\leq 1$, we can get the value for $\eta$ based on \eqref{E:eq4}, i.e., $\eta=\arccos [-(d_1t^2+d_2)/t]$ or $\eta=2\pi-\arccos [-(d_1t^2+d_2)/t]$. Then $\theta$ can be obtained from \eqref{E:thetaEta}. If no such solutions exist, then $\Theta_{\mathcal{A}}$ is set to empty. 

\bibliographystyle{IEEEtran}
\bibliography{IEEEabrv,IEEEfull}

\end{document}